\documentclass[runningheads,envcountsame,a4paper]{llncs}
\usepackage{amsmath,amssymb}

\usepackage{tikz}
\usepackage{url}
\usepackage{float}
\usepackage{array}

\usepackage{graphicx}

\DeclareMathOperator{\cexp}{cexp}
\DeclareMathOperator{\pexp}{pexp}
\DeclareMathOperator{\rtc}{RTC}
\DeclareMathOperator{\rt}{RT}

\DeclareMathOperator{\pref}{pref}
\DeclareMathOperator{\suf}{suff}

\newcommand{\cons}{22}

\frontmatter

\title{Repetition Avoidance in Circular Factors}

\author{Hamoon Mousavi \and Jeffrey Shallit}

\institute{School of Computer Science, University of Waterloo,
Waterloo, ON  N2L 3G1 Canada  \\
\email{ \{sh2mousa,shallit\}@uwaterloo.ca}
}

\begin{document}

\maketitle

\begin{abstract}
We consider the following novel variation on a classical avoidance problem from
combinatorics on words:  instead of avoiding repetitions in all factors
of a word, we avoid repetitions in all factors where each individual
factor is considered as a ``circular word'', i.e., the end of the word wraps around to 
the beginning.  We determine the best
possible avoidance exponent for alphabet size $2$ and $3$, and provide
a lower bound for larger alphabets.
\end{abstract}

\section{Introduction}

Repetition in words is an active research topic and has
been studied for over a hundred years.
For example, Axel Thue \cite{thue06,thue12} constructed an infinite
word over a three-letter alphabet that contains no squares (i.e., no 
nonempty word of the form $xx$),
and another infinite word over a two-letter
alphabet that contains no cubes (i.e., no nonempty word of the form $xxx$).

In 1972, Dejean refined these results by considering fractional powers.
An {\em $\alpha$-power} for a rational number $\alpha \geq 1$ is a word of
the form $w=x^{\lfloor \alpha \rfloor}x'$,
where $x'$ is a (possibly empty) prefix of $x$ and $|w|=\alpha|x|$.
The word $w$ is a {\em repetition},
with a {\em period} $x$ and an {\em exponent} $\alpha$.  Among all
possible exponents, we let $\exp(w)$ denote the largest one, corresponding
to the shortest period.
For example, the word {\tt alfalfa} has shortest
period {\tt alf} and exponent $7\over 3$.    The {\em critical exponent}
of a word $w$ is the supremum, over all factors $f$ of $w$, of 
$\exp(f)$. We write it as $\exp(w)$.

For a real number $\alpha$, an {\em $\alpha^+$-power} is a $\beta$-power where $\beta > \alpha$.
For example $ababa=(ab)^{\frac{5}{2}}$ is a $2^+$-power. 
A word $w$ is
\begin{itemize}
\item {\em $\alpha^+$-power-free}, if none of the factors of
$w$ is an $\alpha^+$-power; 
\item {\em $\alpha$-power-free} if,
in addition to being $\alpha^+$-power-free, the word $w$ has no
factor that is an $\alpha$-power.   
\end{itemize}
We also say that 
$w$ {\it avoids} $\alpha^+$-powers (resp.,
avoids $\alpha$-powers).
Dejean asked, what is the smallest real number $r$ for which there exist infinite $r^+$-power-free 
words over an alphabet of size $k$?  This quantity is called the
{\it repetition threshold} \cite{brandenburg83}, and is denoted by $\rt(k)$. From results of Thue
we know that $\rt(2) = 2$.
Dejean \cite{dejean72} proved $\rt(3)= \frac{7}{4}$, and
conjectured that
\[
 \rt(k) =
 \begin{cases}
 	\frac{7}{5}, & \text{if } k = 4;\\
 	\frac{k}{k-1}, & \text{if $k > 4$}.
 \end{cases}
\]
This conjecture received much attention in the last forty years, and
its proof was recently completed by
Currie and Rampersad \cite{currie10} and Rao \cite{rao09}, 
independently, based on work of Carpi \cite{carpi} and others.

We consider the following novel variation on Dejean, which we call
``circular $\alpha$-power avoidance''.
We consider each finite factor $x$ of a word $w$, but interpret such a factor
as a ``circular'' word, where the end of the word wraps around to 
the beginning.  Then we consider each factor $f$ of this interpretation
of $x$; for $w$ to be circularly $\alpha$-power-free, each such $f$
must be $\alpha$-power-free.
For example, consider the English word $w = {\tt dividing}$
with factor $x = {\tt dividi}$.  The circular shifts of $x$ 
are
$${\tt dividi, ividid, vididi, ididiv, didivi, idivid},$$
and (for example) the word {\tt ididiv} contains a factor {\tt ididi} that is a
${5 \over 2}$-power.
In fact, $w$ is circularly cubefree and circularly
$({5 \over 2})^+$-power-free.

To make this more precise, we recall the notion of conjugacy.
Two words $x, y$ are {\it conjugate} if one is a cyclic shift of the
other; that is, if there exist words $u, v$ such
that $x = uv$ and $y = vu$.

\begin{definition}
Let $w$ be a finite or infinite word.
The largest circular $\alpha$-power
in a word $w$ is defined to be the supremum of $\exp(f)$ over all factors $f$
of conjugates of factors of $w$.  We write it as $\cexp(w)$.
\label{done}
\end{definition}

Although Definition~\ref{done} characterizes the subject of this paper,
we could have used a different definition, based on the following.

\begin{proposition}
Let $w$ be a finite word or infinite word.
The following are equivalent:
\begin{itemize}

\item[(a)] $s$ is a factor of a conjugate of a factor of $w$;

\item[(b)] $s$ is a prefix of a conjugate of a factor of $w$;

\item[(c)] $s$ is a suffix of a conjugate of a factor of $w$;

\item[(d)] $s = vt$ for some factor $tuv$ of $w$.
\end{itemize}
\label{pone}
\end{proposition}

\begin{proof}
(a) $\implies$ (b):  
Suppose $s = y'' x'$, where $xy$ is a factor of
$w$ and $x = x' x''$ and $y = y' y''$.
Another conjugate of $xy$ is then $y'' x' x'' y'$ with prefix
$y'' x'$.

(b) $\implies$ (c):  Such a prefix $s$ is either of the form
$y'$ or $y x'$, where $xy$ be a factor of $w$
and $x = x' x''$ and $y = y' y''$.  
Considering the conjugate $y'' x  y'$ of $yx$, we get a suffix $y'$,
and consider the conjugate $x'' y x'$ we get a suffix $y x'$.

(c) $\implies$ (d):  Such a suffix $s$ is either of the form
$s = x''$ or $s = y'' x$, where
$xy$ be a factor of $w$
and $x = x' x''$ and $y = y' y''$. 
In the former case, let $t = x''$, $u = v = \epsilon$.  In the latter
case, let $t = x$, $u = y'$, and $v = y''$.  

(d) $\implies$ (a):   Let $tuv$ be a factor of $w$.  Then
$vtu$ is a conjugate of $tuv$, and $vt$ is a factor of it.
\qed
\end{proof}

Let $\Sigma_k = \{0,1,\ldots,k-1\}$. Define $\rtc(k)$, the {\em repetition
threshold for circular factors}, to be the smallest real number $r$ for
which there exist infinite circularly $r^+$-power-free words in
$\Sigma_k$. Clearly we have 
$$\rtc(k) \geq \rt(k).$$
In this paper we prove that $\rtc(2) = 4$ and $\rtc(3) = \frac{13}{4}$. For larger alphabets, we conjecture
that
\[
 \rtc(k) =
 \begin{cases}
 	\frac{5}{2}, & \text{if } k = 4;\\
 	\frac{105}{46}, & \text{if $k = 5$};\\
 	\frac{2k-1}{k-1}, & \text{if $k \geq 6$}.\\
 \end{cases}
\]

Finally, we point out that the quantities we study here are {\it not\/}
closely related to the notion of {\it avoidance in circular words},
studied previously in \cite{aberkane,gorbunova,harju}.

\bigskip

\noindent{\it Acknowledgments.}  We thank the referees for their careful reading of this paper.

\section{Notation}

For a finite alphabet $\Sigma$, let $\Sigma^*$ denote the set of finite
words over $\Sigma$. Let $\Sigma^\omega$ denote the set of right
infinite words over $\Sigma$, and let $\Sigma^\infty=\Sigma^\omega\cup
\Sigma^*$. Let $w = a_0a_1\cdots \in \Sigma^\infty$ be a word. Let
$w[i] = a_i$, and let $w[i..j] = a_i\cdots a_j$.  By convention
we have $w[i] =
\epsilon$ for $i < 0$ and $w[i..j] = \epsilon$ for $i > j$.  Note that
if $x$ is a period of $w$ and $|x|=p$, then $w[i+p] =
w[i]$ for $0 \leq i < |w|-p$.

For a word $x$, let $\pref(x)$ and $\suf(x)$, respectively, denote the
set of prefixes and suffixes of $x$. For words $x,y$, let $x \preceq y$
denote that $x$ is a factor of $y$. Let $x \preceq_p y$ (resp.,
$x \preceq_s y$) denote that $x$ is a prefix (resp., suffix) of $y$.

A morphism $h: \Sigma^* \rightarrow \Phi^*$ is said to be {\it $q$-uniform}
if $|h(a)| = q$ for all $a \in \Sigma$. A morphism is uniform if it is
$q$-uniform for some $q$. The fixed point of a morphism
$h: \Sigma^* \rightarrow \Phi^*$ starting with $a\in \Sigma$, if it
exists, is denoted by $h^\omega(a)$.

In the next section, we prove some preliminary results. We get some
bounds for $\rtc(k)$, and in particular, we prove that
$\rtc(2)=2\rt(2)=4$. In Section~\ref{sec:3letter}, we study the three-letter
alphabet, and we prove that $\rtc(3)=\frac{13}{4}$. Finally, in
Section~\ref{sec:another interpretation}, we give another interpretation for
repetition threshold for circular factors.

\section{Binary Alphabet}\label{sec:pre}
First of all, we prove a bound on $\rtc(k)$. 
\begin{theorem}\label{thm:bound}
	$1+\rt(k) \leq \rtc(k) \leq 2\rt(k)$.
\end{theorem}
\begin{proof}
	Let $r = \rt(k)$. We first prove that $\rtc(k) \leq 2r$. Let $w\in \Sigma_k^\omega$ be an $r^+$-power-free word. We prove that $w$ is circularly $(2r)^+$-power-free. Suppose that $xty \preceq w$, such that $yx$ is $(2r)^+$-power. Now either $y$ or $x$ is an $r^+$-power. This implies that $w$ contains an $r^+$-power, a contradiction.
	
	Now we prove that $1+r \leq \rtc(k)$. Let $l$ be the length of the longest $r$-power-free word over $\Sigma_k$, and let $w\in\Sigma_k^\omega$. Considering the factors of length $n = l+1$ of $w$, we know some factor $f$ must occur infinitely often. This $f$ contains an $r$-power: $z^r$. Therefore $z^rtz$ is a factor of $w$. Therefore $w$ contains a circular $(1+r)$-power. This proves that $1+r \leq \rtc(k)$.
	\qed
\end{proof}

Note that since $\rt(k) > 1$, we have $\rtc(k) > 2$. 

\begin{lemma}\label{lemma:2letter}
	$\rtc(2) \geq 4$.
\end{lemma}
\begin{proof}
Let $w\in \Sigma_2^\omega$ be an arbitrary word. It suffices to prove
that $w$ contains circular $4$-powers. There are two cases: either $00$
or $11$ appears infinitely often, or $w$ ends with a suffix of the form
$(01)^\omega$. In the latter case, obviously there are circular
$4$-powers; in the former there are words of the form $aayaa$ for $a
\in \Sigma_2$ and $y \in \Sigma_2^*$ and hence circular $4$-powers.
\qed
\end{proof}

\begin{theorem}
	$\rtc(2)=4$. 
\end{theorem}
\begin{proof}
	A direct consequence of Theorem~\ref{thm:bound} and Lemma~\ref{lemma:2letter}. 
	\qed
\end{proof}

The Thue-Morse word is an example of a binary word that avoids
circular $4^+$-powers.

\section{Ternary Alphabet}\label{sec:3letter}
Our goal in this section is to show that $\rtc(3)=\frac{13}{4}$. For
this purpose, we frequently use the notion of synchronizing morphism, which was
introduced in Ilie et al.~\cite{ochem05}.

\begin{definition}
A morphism $h:\Sigma^* \rightarrow \Gamma^*$ is said to be {\em
synchronizing} if for all $a,b,c \in \Sigma$ and $s,r \in \Gamma^*$, if
$h(ab) = rh(c)s$, then either $r = \epsilon$ and $a = c$ or $s =
\epsilon$ and $b = c$.
\end{definition}

\begin{definition}
A synchronizing morphism $h:\Sigma^* \rightarrow \Gamma^*$ is said to
be {\em strongly synchronizing} if for all $a,b,c \in \Sigma$, if $h(c)
\in \pref(h(a))\suf(h(b))$, then either $c = a$ or $c = b$.
\end{definition}

The following technical lemma is applied several times throughout the paper.
\begin{lemma}\label{lemma:technicalLemma}
	Let $h:\Sigma^* \rightarrow \Gamma^*$ be a synchronizing $q$-uniform morphism. Let $n > 1$ be an integer, and let $w \in \Sigma^*$. If $z^n \preceq_p h(w)$ and $|z| \geq q$, then $u^n \preceq_p w$ for some $u$. Furthermore $|z| \equiv 0$ (mod $q$). 
\end{lemma}
\begin{proof}
	Let $z = h(u)z'$, where $|z'| < q$ and $u \in \Sigma^*$. Note that $u \neq \epsilon$, since $|z| \geq q$. Clearly, we have $z'h(u[0]) \preceq_p h(w[|u|..|u|+1])$. Since $h$ is synchronizing, the only possibility is that $z' = \epsilon$, so $|z| \equiv 0$ (mod $q$). Now we can write $z^n=h(u^n)\preceq_p h(w)$. Therefore $u^n \preceq_p w$.
	\qed
\end{proof}

The next lemma states that if the fixed point of a strongly
synchronizing morphism (SSM) avoids small $n$'th powers, where $n$ is an
integer, it avoids
$n$'th powers of all lengths.

\begin{lemma}\label{lemma:powersInSSM}
	Let $h: \Sigma^* \rightarrow \Sigma^*$ be a strongly synchronizing $q$-uniform morphism. Let $n > 1$ be an integer. If $h^\omega(0)$ avoids factors of the form $z^n$, where $|z^n| < 2nq$, then $h^\omega(0)$ avoids $n$'th powers.
\end{lemma}
\begin{proof}
	Let $w = a_0a_1a_2\cdots = h^\omega(0)$. Suppose $w$ has $n$'th powers of length greater than or equal to $2nq$. Let $z$ be the shortest such word, i.e., $|z^n| \geq 2nq$ and $z^n \preceq w$. We can write 
	\begin{align*}
		&z^n = xh(w[i..j])y,\\
		&x \preceq_s h(a_{i-1}),\\
		&y \preceq_p h(a_{j+1}),\\
		&|x|,|y| < q,
	\end{align*}
	for some integers $i,j \geq 0$. If $x=y=\epsilon$, then using Lemma~\ref{lemma:technicalLemma}, since $|z| \geq q$, the word $w[i..j]$ contains an $n$'th power. Therefore $w$ contains an $n$'th power of length smaller than $|z^n|$, a contradiction. Now suppose that $xy \neq \epsilon$. Since $|z| \geq \frac{2nq}{n} = 2q$, and $|xh(w[i])|,|h(w[j])y| < 2q$, we can write 
	\begin{align*}
		&xh(w[i]) \preceq_p z, \\
		&h(w[j])y \preceq_s z.
	\end{align*}
	Therefore $h(w[j])yxh(w[i]) \preceq z^2 \preceq z^n$. Since $h$ is synchronizing $$h(w[j])yxh(w[i]) \preceq h(w[i..j]).$$ Hence $yx = h(a)$ for some $a \in \Sigma$. Since $h$ is an SSM, we have either $a = a_{i-1}$ or $a = a_{j+1}$. Without loss of generality, suppose that $a=a_{i-1}$. Then we can write $h(w[i-1..j])=yxh(w[i..j])$. The word $yxh(w[i..j])$ is an $n$'th power, since it is a conjugate of $xh(w[i..j])y$. So we can write $$h(w[i-1..j]) = z_1^n$$ where $z_1$ is a conjugate of $z$. Note that $|z_1| = |z| \geq 2q$. Now using Lemma~\ref{lemma:technicalLemma}, the word $w[i-1..j]$ contains an $n$'th power, and hence $w$ contains an $n$'th power of length smaller than $|z^n|$, a contradiction. 
	\qed
\end{proof}

The following lemma states that, for an SSM $h$ and a well-chosen word $w$, all circular $(\frac{13}{4})^+$-powers in $h(w)$ are small.

\begin{lemma}\label{lemma:mainlemma}
Let $h: \Sigma^* \rightarrow \Gamma^*$ be a strongly synchronizing
$q$-uniform morphism. Let $w=a_0a_1a_2\cdots \in \Sigma^\omega$ be a
circularly cubefree word. In addition, suppose that $w$ is squarefree.
If $x_1tx_2 \preceq h(w)$ for some words $t,x_1,x_2$, and $x_2x_1$ is
a $(13/4)^+$-power, then $|x_2x_1| < \cons q$.
\end{lemma}
\begin{proof}
	The proof is by contradiction. Suppose there are words $t,x_1,x_2$, and $z$ in $\Gamma^*$ and a rational number $\alpha > \frac{13}{4}$ such that 
	\begin{equation*}
		x_1tx_2 \preceq h(w)
	\end{equation*}
	\begin{equation*}
		|x_2x_1| \geq \cons q
	\end{equation*}
	\begin{equation*}
		x_2x_1 = z^\alpha.
	\end{equation*}
	
	Suppose $|z| < q$. Let $k$ be the smallest integer for which $|z^k| \geq q$. Then $|z^k| < 2q$, because otherwise $|z^{k-1}| \geq q$, a contradiction. We can write $x_2x_1 = (z^k)^\beta$, where $\beta = \frac{|x_2x_1|}{|z^k|} > \frac{22q}{2q} > \frac{13}{4}$. Therefore we can assume that $|z| \geq q$, since otherwise we can always replace $z$ with $z^k$, and $\alpha$ with $\beta$. 
		
	There are three cases to consider. 
	\begin{enumerate}
		\item[(a)]	Suppose that $x_1$ and $x_2$ are both long enough, so that each contains an image of a word under $h$. More formally, suppose that
		\begin{align}
			&x_1 = y_1h(w[i_1..j_1])y_2,\label{eq:x1}\\
			&x_2 = y_3h(w[i_2..j_2])y_4,\label{eq:x2}\\
			&i_1 \leq j_1, i_2 \leq j_2,\notag\\
			&y_1 \preceq_s h(a_{i_1-1}),\notag\\
			&y_2 \preceq_p h(a_{j_1+1}),\notag\\
			&y_3 \preceq_s h(a_{i_2-1}),\notag\\
			&y_4 \preceq_p h(a_{j_2+1}),\notag\\
			&|y_1|,|y_2|,|y_3|, \text{ and } |y_4| < q, \text{ and}\notag\\
			&y_2ty_3 = h(w[j_1+1..i_2-1]).\notag
		\end{align}
Let $v_1 = w[i_1..j_1]$ and $v_2 = w[i_2..j_2]$.
See Fig~\ref{fig1}.
	
\begin{figure}[H]
\begin{tikzpicture}[scale=0.75]
				
\begin{scope}
				\draw (3.5,.25) node[left]{$w=$}++(0,-.25) rectangle +(1,.5)
		 	  			++(1,0) rectangle node{$v_1$}+(2,.5)
	      				++(2,0) rectangle +(2.3,.5)
	      				++(2.3,0) rectangle node{$v_2$}+(2,.5)
	      				++(2,0) rectangle +(1,.5);
			\end{scope}
			
			\draw[->,semithick,dotted] (4.5,0) node[below]{$i_1$} -- (4,-1.5);
			\draw[->,semithick,dotted] (6.5,0) node[below]{$j_1$} -- (7,-1.5);
			
			\draw[->,semithick,dotted] (8.8,0) node[below]{$i_2$} -- (11,-1.5);
			\draw[->,semithick,dotted] (10.8,0) node[below]{$j_2$}-- (14,-1.5);
			
			\begin{scope}[yshift=-2cm]
				\draw (2,.25)node[left]{$h(w)=$}++(0,-.25) rectangle +(1,.5)
	      				++(1,0) rectangle node{$y_1$}+(1,.5)
 			  			++(1,0) rectangle node{$h(w[i_1..j_1])$}+(3,.5)
 			  			++(3,0) rectangle node{$y_2$}+(1,.5) 
  				  		++(1,0) rectangle node{$t$}+(2,.5)
  				  		++(2,0) rectangle node{$y_3$}+(1,.5)
  				  		++(1,0) rectangle node{$h(w[i_2..j_2])$}+(3,.5)
  				  		++(3,0) rectangle node{$y_4$}+(1,.5)
  				  		++(1,0) rectangle +(1,.5);
  				  
	   			\draw (3,-.5) rectangle node{$x_1$} (8,0);
	   			\draw (8,-.5) rectangle node{$t$} (10,0);
	   			\draw (10,-.5) rectangle node{$x_2$} (15,0);
			\end{scope}
		\end{tikzpicture}
		\caption{$x_1tx_2$ is a factor of $h(w)$} \label{fig1}
\end{figure}
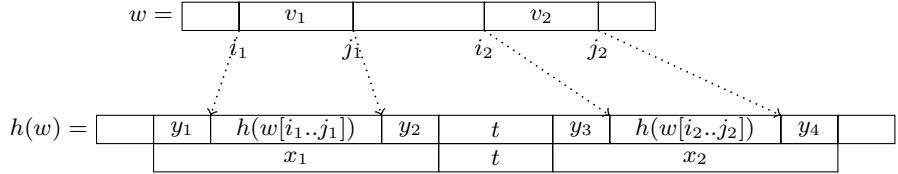
	There are two cases to consider.
	\begin{enumerate}
		\item[(1)] Suppose that $y_4y_1 = \epsilon$. Let $v = w[i_2..j_2]w[i_1..j_1]$.

		The word $h(v)y_2$ is a factor of $y_3h(v)y_2 = z^\alpha$ of length $\geq 22q - q = 21q$, and so $$h(v)y_2 = z_1^\beta,$$ where $z_1$ is a conjugate of $z$, and $\beta \geq \frac{21}{22}\alpha > 3$. Therefore we can write $$z_1^3 \preceq_p h(v)y_2 \preceq_p h(vw[j_1+1]).$$ Note that $|z_1| = |z| \geq q$, so using Lemma~\ref{lemma:technicalLemma}, we can write $|z_1| \equiv 0$ (mod $q$). Therefore $$z_1^3 \preceq_p h(v).$$ 
		Using Lemma~\ref{lemma:technicalLemma} again, the word
		$v$ contains a cube, which means that the word $w$
		contains a circular cube, a contradiction.
			
		\item[(2)] Suppose that $y_4y_1 \neq \epsilon$. We show how to get two new factors $x'_1=h(v'_1)y'_2$ and $x'_2=y'_3h(v'_2)$, with $v'_1,v'_2$ nonempty, such that $x'_2x'_1=x_2x_1$. Then we use case (1) above to get a contradiction.
		
		Let $s = h(w[j_2])y_4y_1h(w[i_1])$, and let $m$ be the smallest integer for which $|z^m| \geq |s|$. Note that if $|z| < |s|$, then \begin{equation}
			|z^m| < 2|s| < 8q. \label{eq:2s8q}
		\end{equation} 
		We show that at least one of the following inequalities holds: 
		\begin{align*}
			|h(v_1)| &\geq q + |z^m|,\\
			|h(v_2)| &\geq q + |z^m|.
		\end{align*}
		Suppose that both inequalities fail. Then using (\ref{eq:x1}) and (\ref{eq:x2}) we can write
		\begin{equation}
			|x_2x_1| < 2q + 2|z^m| + |y_1y_2y_3y_4| < 6q + 2|z^m|. \label{eq:x2x1}
		\end{equation} 
		If $|z| < |s|$, then by (\ref{eq:2s8q}) and (\ref{eq:x2x1}) one gets $|x_2x_1| < 22q$, contradicting our assumption. Otherwise $|z| \geq |s|$, and hence $m=1$. Then $$|x_2x_1| = \alpha|z| < 2q + 2|z| + |y_1y_2y_3y_4| < 6q + 2|z|,$$ and hence $|z| < 6q$. So $|x_2x_1| < 6q + 2|z| < 18q$, contradicting our assumption. Without loss of generality, suppose that $|h(v_1)| \geq q + |z^m|$.
		
		Using the fact that $z$ is a period of $x_2x_1$, we can write
\begin{equation*}
h(v_1)[q + |z^m| - |s|..q + |z^m| - 1] = s,
\end{equation*}
		or, in other words,
\begin{equation*}
s \preceq h(v_1).  
\end{equation*}
See Fig~\ref{fig2}.
		
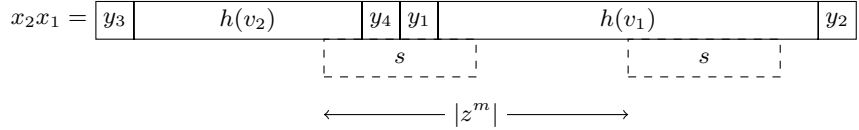
\begin{figure}[H]
\begin{tikzpicture}[scale=1.0]
\begin{scope}
			\draw (1,.25) node[left]{$x_2x_1=$}++(0,-.25) rectangle node{$y_3$}+(.5,.5)
		 			++(.5,0) rectangle node{$h(v_2)$}+(3,.5)
		   			++(3,0) rectangle node{$y_4$}+(.5,.5)
			  		++(.5,0) rectangle node{$y_1$}+(.5,.5)
		  	  		++(.5,0) rectangle node{$h(v_1)$}+(5,.5)
		 	  		++(5,0) rectangle node{$y_2$}+(.5,.5);
		 	\draw[dashed] (4,0) rectangle node{$s$}(6,-.5);
		 	\draw[<->] (4,-1) --node[fill=white]{$|z^m|$} (8,-1);
		 	\draw[dashed] (8,0) rectangle node{$s$}(10,-.5);
		\end{scope}
		\end{tikzpicture}
		\caption{$h(v_1)$ contains a copy of $s$} \label{fig2}
		\end{figure}

		Using the fact that $h$ is synchronizing, we get that $y_4y_1 = h(a)$ for some $a \in \Sigma$. Since $h$ is an SSM, we have either $a = a_{i_1-1}$ or $a = a_{j_2+1}$. Without loss of generality, suppose that $a = a_{j_2+1}$. 
		Now look at the following factors of $h(w)$, which can be obtained from $x_1$ and $x_2$ by extending $x_2$ to the right and shrinking $x_1$ from the left:
		\begin{align*}
			x'_1 &= h(w[i_1..j_1])y_2\\
			x'_2 &= y_3h(w[i_2..j_2+1]). 
		\end{align*}
See Fig~\ref{fig3}. 
\begin{figure}[H]
\begin{tikzpicture}[scale=0.77]
\begin{scope}
				\draw[dashed] (2.25,.5) rectangle node{$x_1$} +(5,.5);	   			
				\draw[dashed] (8.75,.5) rectangle node{$x_2$} +(4.75,.5);
			
				\draw (1,.25)node[left]{$h(w)=$}++(0,-.25) rectangle +(.5,.5)
	      				++(.5,0) rectangle node{$h(a_{i_1-1})$}+(1.75,.5)
 			  			++(1.75,0) rectangle node{$h(a_{i_1}a_{i_1+1}\cdots a_{j_1})$}+(3.25,.5)
 			  			++(3.25,0) rectangle node{$y_2$}+(.75,.5) 
  				  		++(.75,0) rectangle node{$t$}+(1.5,.5)
  				  		++(1.5,0) rectangle node{$y_3$}+(.75,.5)
  				  		++(.75,0) rectangle node{$h(a_{i_2}a_{i_2+1}\cdots a_{j_2})$}+(3.25,.5)
  				  		++(3.25,0) rectangle node{$h(a_{j_2+1})$}+(1.75,.5)
  				  		++(1.75,0) rectangle +(.5,.5);

	   			\draw[dashed] (3.25,-.5) rectangle node{$x'_1$} +(4,.5);
	   			\draw[dashed] (8.75,-.5) rectangle node{$x'_2$} +(5.75,.5);
			\end{scope}
		\end{tikzpicture}
\caption{$x'_1$ and $x'_2$ are obtained from $x_1$ and $x_2$} 
\label{fig3}
We can see that $$x'_2x'_1=y_3h(w[i_2..j_2+1])h(w[i_1..j_1])y_2=y_3h(w[i_2..j_2])y_4y_1h(w[i_1..j_1])y_2=x_2x_1.$$
Now using case (1) we get a contradiction.
\end{figure}
			\end{enumerate}
		\item[(b)] Suppose that $x_2$ is too short
to contain an image of a word under $h$. More formally, we can write
$$ x_1 = y_1h(v)y_2 \text{ where } |x_2| < 2q \text{ and } |y_1|,|y_2|< q $$
for some words $y_1,y_2 \in \Gamma^*$ and a word $v \preceq w$. Then $h(v)$ is a factor of $x_2x_1 = z^\alpha$ of length $\geq 22q - 4q = 18q$, and so $$h(v) = z_1^\beta,$$ where $z_1$ is a conjugate of $z$, and $\beta \geq \frac{18}{22}\alpha > 2$. Note that $|z_1|=|z|\geq q$, so using Lemma~\ref{lemma:technicalLemma}, the word $v$ contains a square, a contradiction.
		
		\item[(c)] Suppose that $x_1$ is not long enough to contain an image of a word under $h$. An argument similar to (b) applies here to get a contradiction.
	\end{enumerate}	
	\qed
\end{proof}

The following 15-uniform morphism is an example of an SSM:
\begin{align*}
	\mu(0) &= 012102120102012\\
	\mu(1) &= 201020121012021\\
	\mu(2) &= 012102010212010\\
	\mu(3) &= 201210212021012\\
	\mu(4) &= 102120121012021\\
	\mu(5) &= 102010212021012.\\
\end{align*}

Another example of an SSM is the $4$-uniform morphism $\psi: \Sigma_6^* \rightarrow \Sigma_6^*$ as follows:
\begin{align*}
	\psi(0) &= 0435\\
	\psi(1) &= 2341\\
	\psi(2) &= 3542\\
	\psi(3) &= 3540\\
	\psi(4) &= 4134\\
	\psi(5) &= 4105.\\
\end{align*}

Our goal is to show that $\mu(\psi^\omega(0))$ is circularly $(\frac{13}{4})^+$-power-free. For this purpose, we first prove that $\psi^\omega(0)$ is circularly cubefree. Then we apply Lemma~\ref{lemma:mainlemma}, for $h = \mu$ and $w = \psi^\omega(0)$.

\begin{lemma}\label{lemma:squarefreeness}
	The fixed point $\psi^\omega(0)$ is squarefree. 
\end{lemma}
\begin{proof}
	Suppose that $\psi^\omega(0)$ contains a square. Using Lemma~\ref{lemma:powersInSSM}, there is a square $zz \preceq \psi^\omega(0)$ such that $|zz| < 16$. Using a computer program, we checked all factors of length smaller than $16$ in $\psi^\omega(0)$, and none of them is a square. This is a contradiction.
	\qed
\end{proof}

\begin{lemma}\label{lemma:circularcubeLemma}
	The fixed point $\psi^\omega(0)$ is circularly cubefree.
\end{lemma}
\begin{proof}
	By contradiction. Let $w = a_0a_1a_2 \cdots = \psi^\omega(0)$. Suppose $x_1tx_2 \preceq w$, and $x_2x_1 = z^3$ for some words $t,x_1,x_2,z$, where 	
	\begin{align*}
		&x_1 = y_1\psi(w[i_1..j_1])y_2,\\
		&x_2 = y_3\psi(w[i_2..j_2])y_4,\\
		&y_1 \preceq_s \psi(a_{i_1-1}),\\
		&y_2 \preceq_p \psi(a_{j_1+1}),\\
		&y_3 \preceq_s \psi(a_{i_2-1}),\\
		&y_4 \preceq_p \psi(a_{j_2+1}),\\
		&|y_1|,|y_2|,|y_3|, \text{ and } |y_4| < 4,\\
		&y_2ty_3 = \psi(w[j_1+1..i_2-1]),
	\end{align*}
	for proper choices of the integers $i_1,i_2,j_1,j_2$. Let $v_1 = w[i_1..j_1]$ and $v_2 = w[i_2..j_2]$. 
	
	Using a computer program, we searched for circular cubes of length not greater than $66$, and it turns out that there is no such circular cube in $w$. Thus we can assume that $|x_2x_1| > 66$ so $|z| > 22$. Moreover suppose that $x_2x_1$ has the smallest possible length. 
	
	There are two cases to consider.
	\begin{enumerate}
	
	\item[(a)] Suppose that $y_4y_1 = \epsilon$.
	If $y_2y_3 = \epsilon$, then $\psi(v_2v_1) = z^3$. Using Lemma~\ref{lemma:technicalLemma}, we get that $v_2v_1$ contains a cube. Hence $w$ contains
	a smaller circular cube, a contradiction. 
		
		Suppose that $y_2y_3 \neq \epsilon$. Since $|y_3\psi(w[i_2])|,|\psi(w[j_1])y_2| < 8$ and $|z| > 22$, we can write 
		\begin{align*}
			&y_3\psi(w[i_2]) \preceq_p z, \\
			&\psi(w[j_1])y_2 \preceq_s z.
		\end{align*}
		Therefore $\psi(w[j_1])y_2y_3\psi(w[i_2]) \preceq z^3$, and since $\psi$ is synchronizing $$\psi(w[j_1])y_2y_3\psi(w[i_2]) \preceq \psi(v_2v_1).$$ Hence $y_2y_3 = \psi(b)$ for some $b \in \Sigma_6$. Since $\psi$ is an SSM, we have either $b = a_{i_2-1}$, or $b = a_{j_1+1}$. Without loss of generality, suppose that $b=a_{i_2-1}$. So we can write $$\psi(w[i_2-1..j_2]w[i_1..j_1]) = y_2y_3\psi(w[i_2..j_2]w[i_1..j_1]).$$ The word $y_2y_3\psi(v_2v_1)$ is a cube, since it is a conjugate of $y_3\psi(v_2v_1)y_2$. So we can write $$\psi(w[i_2-1..j_2]w[i_1..j_1])=z_1^3$$ where $z_1$ is a conjugate of $z$. Then using Lemma~\ref{lemma:technicalLemma}, the word $w[i_2-1..j_2]w[i_1..j_1]$ contains a cube. Note that since $y_2y_3 \neq \epsilon$ we have $j_1 < i_2-1$. Hence $w[i_2-1..j_2]w[i_1..j_1]$ is a circular cube of $w$, a contradiction.
	
	\item[(b)] Suppose that $y_4y_1 \neq \epsilon$. We show how to get two new factors $x'_1=h(v'_1)y'_2$ and $x'_2=y'_3h(v'_2)$ of $w$, for nonempty words $v'_1,v'_2$, such that $x'_2x'_1 = x_2x_1$. Then we use case (a) above to get a contradiction.
	
	The word $w$ is squarefree due to Lemma~\ref{lemma:squarefreeness}. Therefore $|x_1|,|x_2| > |z| > \frac{66}{3}$ and hence $|v_1|,|v_2| > 0$. One can observe that either $|\psi(v_1)| \geq 4 + |z|$ or $|\psi(v_2)| \geq 4 + |z|$. Without loss of generality, suppose that $|\psi(v_1)| \geq 4 + |z|$. Let $s = w[j_2]y_4y_1w[i_1]$. Now, using the fact that $z$ is a period of $x_2x_1$,
we can write
$$\psi(v_1)[4 + |z| - |s|..4 + |z| - 1] = s,$$ 
or, in other words, 
$$s \preceq \psi(v_1).$$
	
	Using the fact that $\psi$ is synchronizing, we get that $y_4y_1 = \psi(a)$ for some $a \in \Sigma_6$. Since $\psi$ is an SSM, we have either $a = a_{i_1-1}$, or $a = a_{j_2+1}$. Without loss of generality, suppose that $a = a_{j_2+1}$. 
	Now look at the following factors of $w$, which can be obtained from $x_1$ and $x_2$ by extending $x_2$ to the right and shrinking $x_1$ from the left
	\begin{align*}
		x'_1 &= \psi(w[i_1..j_1])y_2\\
		x'_2 &= y_3\psi(w[i_2..j_2+1]). 
	\end{align*}
	We can write
	\begin{equation*}
		x'_2x'_1 = y_3\psi(w[i_2..j_2+1])\psi(w[i_1..j_1])y_2 = y_3\psi(v_2)y_4y_1\psi(v_1)y_2 = x_2x_1 = z^3\label{eq:(x'2x'1)thesecondtime}.
	\end{equation*} So using case (a) we get a contradiction. 
	\end{enumerate}
	\qed
\end{proof}

\begin{theorem}
	$\rtc(3) = \frac{13}{4}$.
\end{theorem}	
\begin{proof}
	First let us show that $\rtc(3) \geq \frac{13}{4}$. 
	
Suppose there exists an infinite word $w$ that avoids circular
$\alpha$-powers, for $\alpha < 4$. We now argue that for every integer $C$,
there exists an infinite word $w'$ that avoids both squares of length $\leq C$
and
circular $\alpha$-powers. Note that none of the factors of $w$ looks
like $xxyxx$, since $w$ avoids circular $4$-powers. Therefore, every
square in $w$ occurs only finitely many times. Therefore $w'$ can be
obtained by removing a sufficiently long prefix of $w$.

Computer search verifies that the longest circularly $\frac{13}{4}$-power-free word over a $3$-letter
alphabet that avoids squares $xx$ where $|xx| < 147$ 
has length $147$. Therefore the above argument for $C = 147$ shows that circular $\frac{13}{4}$-powers are
unavoidable over a $3$-letter alphabet.

Now to prove $\rtc(3) = \frac{13}{4}$, it is sufficient to give an example
of an infinite word that avoids circular $(\frac{13}{4})^+$-powers. We
claim that $\mu(\psi^\omega(0))$ is such an example. We know that
$\psi^\omega(0)$ is circularly cubefree. Therefore we can use Lemma
\ref{lemma:mainlemma} for $w = \psi^\omega(0)$ and $h = \mu$. So if $xty
\preceq \mu(\psi^\omega(0))$, and $yx$ is a $(\frac{13}{4})^+$-power,
then $|yx| < \cons \times 15$.
Now there are finitely many possibilities for $x$ and $y$. 
Using a computer program, we checked that none of them leads to a $(\frac{13}{4})^+$-power.
This completes the proof.
\qed	
\end{proof}

\section{Another Interpretation}\label{sec:another interpretation}

We could, instead, consider the
supremum of $\exp(p)$ over all products of
$i$ factors of $w$.    Call this quantity $\pexp_i (w)$.  

\begin{proposition}\label{lemma:recurrent words}
If $w$ is a recurrent infinite word, then $\pexp_2 (w) = \cexp (w)$.
\end{proposition}

\begin{proof}
Let $s$ be a product of two factors of $w$,
say $s = xy$.
Let
$y$ occur for the first time at position $i$ of $w$.  Since $w$ is
recurrent, $x$ occurs somewhere after position $i + |y|$ in $w$.
So there exists $z$ such that $yzx$ is a factor of $w$.  Then
$xy$ is a factor of a conjugate of a factor of $w$.

On the other hand, from Proposition~\ref{pone}, we know that if 
$s$ is a conjugate of a factor of $w$, then $s = vt$ where
$tuv$ is a factor of $w$.  Then $s$ is the product of two factors of $w$.
\qed
\end{proof}

We can now study the repetition threshold for $i$-term products,
$\rt_i (k)$,
which is
the infimum of $\pexp_i (w)$ over all words $w \in \Sigma_k^\omega$. Note that $$\rt_2(k) \geq \rtc(k).$$

The two recurrent words, the Thue-Morse word and $\mu(\psi^\omega(0))$, introduced in 
Section \ref{sec:3letter}, are circularly $\rtc(2)^+$-power-free 
and circularly $\rtc(3)^+$-power-free, respectively. 
Using Proposition \ref{lemma:recurrent words}, we get that $\rt_2(k)=\rtc(k)$ for $k=2,3$.

\begin{theorem} For $i \geq 1$ we have
$\rt_i (2) = 2i$.
\end{theorem}

\begin{proof}
From Thue we know there exists an infinite $2^+$-power-free word.
If some product of factors $x_1 x_2 \cdots x_i$ contains a $(2i)^+$-power, then
some factor contains a $2^+$-power, a contradiction.  So $\rt_i (2) \leq 2i$.

For the lower bound, fix $i \geq 2$, and let $w \in \Sigma_2^\omega$ be
an arbitrary word.  Either $00$ or $11$ appears infinitely often, or 
$w$ ends in a suffix of the form $(01)^\omega$.  In the latter case we
get arbitrarily high powers, and the former case there is
a product of $i$ factors with exponent $2i$.
\qed
\end{proof}

It would be interesting to obtain more values of $\rt_i(k)$.
We propose the following conjectures which are supported by numerical evidence:
\begin{align*}
	&\rt_2(4) = \rtc(4) = \frac{5}{2} \text { , }\\ 
	&\rt_2(5) = \rtc(5) = \frac{105}{46} \text { , and }\\ 
	&\rt_2(k) = \rtc(k) = 1 + \rt(k) = \frac{2k-1}{k-1} \text{ for } k \geq 6.\\
\end{align*}
We know that the values given above are lower bounds for $\rtc(k)$.


\begin{thebibliography}{99}

\bibitem{aberkane}
A. Aberkane and J. D. Currie. There exist binary circular ${5/2}^+$ power free 
words of every length.  \emph{Electronic J. Combin.} {\bf 11} (1), 2004, 
Paper \#R10.  Available at \url{http://www1.combinatorics.org/Volume_11/Abstracts/v11i1r10.html}.

\bibitem{brandenburg83}
F.-J. Brandenburg. Uniformly growing k-th power-free homomorphisms.
\emph{Theoret. Comput. Sci.} {\bf 23} (1983), 69--82.

\bibitem{carpi}
A. Carpi. On Dejean's conjecture over large alphabets. \emph{Theoret. Comput. Sci.} {\bf 385} (2007),
137--151.

\bibitem{currie10}
J. Currie and N. Rampersad. A proof of Dejean's conjecture. \emph{Math. Comp.} {\bf 80} (2011), 1063--1070.

\bibitem{dejean72}
F. Dejean. Sur un th\'{e}or\`{e}me de Thue. \emph{J. Combin. Theory. Ser. A} {\bf 13} (1972), 90--99.

\bibitem{gorbunova}
I. A. Gorbunova. Repetition threshold for circular words.
{\it Electronic J. Combin.} {\bf 19} (4), Paper \#11.
Available at \url{http://www.combinatorics.org/ojs/index.php/eljc/article/view/v19i4p11}.

\bibitem{harju}
T. Harju and D. Nowotka. Cyclically repetition-free words on small alphabets.
\emph{Inform Process Lett.} {\bf 110} (2010) 591--595.

\bibitem{ochem05}
L. Ilie, P. Ochem, and J. Shallit. A generalization of repetition threshold. \emph{Theoret. Comput. Sci.} {\bf 345} (2005), 359--369.

\bibitem{rao09}
M. Rao. Last cases of Dejean's conjecture. \emph{Theoret. Comput. Sci.} {\bf 412} (2011), 3010--3018.

\bibitem{thue06}
A. Thue. \"{U}ber unendliche Zeichenreihen.
\emph{Norske vid. Selsk. Skr. Mat. Nat. Kl.}
{\bf 7} (1906), 1--22. Reprinted in \emph{Selected Mathematical 
Papers of Axel Thue,} T.
Nagell, ed., Universitetsforlaget, Oslo, 1977, pp.\ 139--158.

\bibitem{thue12}
A. Thue. \"{U}ber die gegenseitige Lage gleicher Teile gewisser Zeichen
reihen.  \emph{Norske vid. Selsk. Skr. Mat. Nat. Kl.} {\bf 1} (1912),
1--67. Reprinted in \emph{Selected Mathematical Papers of Axel Thue,}
T. Nagell, ed., Universitetsforlaget, Oslo, 1977, pp.\ 413--478.

\end{thebibliography}
\end{document}